\documentclass[a4paper]{article}
\usepackage{amsmath}
\usepackage{amsthm}
\usepackage{amssymb}

\newtheorem{theorem}{Theorem}[section]
\newtheorem{corollary}[theorem]{Corollary}
\newtheorem{definition}[theorem]{Definition}

\newtheorem{lemma}[theorem]{Lemma}

\theoremstyle{remark}

\DeclareMathOperator{\supp}{supp}

\title{Well-Posedness and Comparison Principle for Option Pricing with Switching Liquidity}
\author{T.B. Gyulov and L.G. Valkov\thanks{e-mails: tgulov@uni-ruse.bg (T.Gyulov) and lvalkov@uni-ruse.bg (L.Valkov)
\vspace{6pt}}\\\vspace{6pt}  {\small\em{University of Ruse ``Angel Kanchev'', 7017 Ruse, Bulgaria}}}

\begin{document}
\maketitle
\begin{abstract}
We consider an integro-differential equation derived from a  system of coupled parabolic PDE and an ODE which describes an European option pricing with liquidity shocks. We study the well-posedness and prove comparison principle for the corresponding initial value problem.
\end{abstract}

\section{Introduction}
This work is devoted to the study of an initial value problem of the following form
\begin{equation}
\left\{\begin{array}{rl}\displaystyle
	\frac{\partial u}{\partial \tau}-\frac{1}{2}\sigma^2S^2\frac{\partial^2 u}{\partial S^2}&=-\nu_{01}e^{u(S,\tau )}\left(\displaystyle\nu_{10}\int_0^\tau e^{-u(S,s)}ds+e^{-\gamma h(S)}\right)+\kappa,\\
	u\left(S,0\right)&=\gamma h(S).
\end{array}\right.
	\label{eq:psintdiff}
\end{equation}
Here $\tau\in[0,T]$, $S\in(0,+\infty)$, $h(S)$ is a given function and  $\sigma$, $\nu_{01}$, $\nu_{10}$, $\kappa$ and $\gamma$ are constants.

The integro-differential equation in \eqref{eq:psintdiff} is derived from a system of coupled parabolic PDE and ODE which is suggested  by M. Ludkovski and Q. Shen \cite{LudShen} in European option pricing in a financial market switching between two states -a liquid state (0) and an illiquid (1) one.  We briefly describe their model. First, it is assumed that the dynamics of the liquidity  is represented by a continuous-time Markov chain $(M_t)$ 
with intensity rates of the transitions $0\rightarrow1$ and $1\rightarrow0$ and determined by the constants $\nu_{01}$ and $\nu_{10}$, respectively.
During the liquid phase $(M_t=0)$ the market dynamics follows the classical Black-Scholes model. More precisely, the price $S_t$ of a stock is modelled by geometric Brownian motion
\[dS_t=\mu S_tdt+\sigma S_tdW_t,\]
with drift $\mu$ and volatility $\sigma$ and a standard one-dimensional Brownian motion $(W_t)$ which is independent of the Markov chain $(M_t)$ (under the ``real world'' probability $\mathbb{P}$). Then the wealth process $(X_t)$ satisfies \[dX_t=\mu\pi_tX_tdt+\sigma\pi_tX_tdW_t,\] where $\pi_t$ denotes  the proportion of stock holdings in the total wealth $X_t$. For simplicity, it is assumed that the interest rate of the riskless asset is zero. 

Respectively, in the illiquid phase $(M_t=1)$, the market is static and trading in stock is not permitted, i.e., $dS_t=dX_t=0$.

The presence of liquidity shocks is a source of non-traded risk and makes the market incomplete. Ludkovski and Shen investigate  expected utility maximization with exponential utility function: 
\[u(x)=-e^{-\gamma x}, \] 
where $\gamma>0$ is the investor's risk aversion parameter. The value functions $\hat{U}^i(t,X,S)$, $i=0,1$ for the optimal investment problem are defined as follows:
\[\hat{U}^i(t,X,S):=\sup_{\pi_t}\mathbb{E}^{\mathbb{P}}_{t,X,S,i}\left[-e^{-\gamma\left(X_T+h(S_T)\right)}\right], \quad i=0,1,\]
where $\mathbb{E}^{\mathbb{P}}_{t,X,S,i}$ is the expectation under the measure $\mathbb{P}$ with starting values $S_t=S$, $X_t=X$ and $M_t=i$. The supremum above is taken over all admissible trading strategies $(\pi_t)$ and the function $h(S)$ denotes the terminal payoff of a contingent claim. 
Standard stochastic control methods and the properties of the exponential utility function imply that the value functions can be presented by 
\[\hat{U}^i(t,X,S)=-e^{-\gamma X}e^{-\gamma R^i(t,S)}, \quad i=0,1,\] 
where $R^i(t,S)$ are the unique viscosity solutions of the system (\cite{LudShen})
\begin{equation}\label{R01systemPDE}
\left\{
\begin{array}{l}
R^0_t+\frac{1}{2}\sigma^2S^2R^0_{SS}-\frac{\nu_{01}}{\gamma}e^{-\gamma(R^1-R^0)}+\frac{d_0+\nu_{01}}{\gamma}=0,\\
R^1_t-\frac{\nu_{10}}{\gamma}e^{-\gamma(R^0-R^1)}+\frac{\nu_{10}}{\gamma}=0,	
\end{array}
\right.
\end{equation}
with the terminal condition $R^i(T,S)=h(S)$, $i=0,1$. Here  $d_0:=\mu^2/2\sigma^2$.

Let  $p$ and $q$ denote the buyer's indifference prices corresponding to liquid and illiquid initial state respectively. They 
are defined as follows: $\hat{U}^0(t,X-p,S)=\hat{V}^0(t,X)$ and $\hat{U}^1(t,X-q,S)=\hat{V}^1(t,X)$ where $\hat{V}^i$, $i=0,1$ are the value functions of the Merton optimal investment problem (i.e. the case when $h(S)\equiv0$). It can be shown that $p$ and $q$ satisfy a system of differential equations which is quite similar to \eqref{R01systemPDE} (see \eqref{eq:pqsyst}). In fact, \[p=R^0+\gamma^{-1}\ln F_0(t)\qquad \text{and} \qquad q=R^1+\gamma^{-1}\ln F_1(t)\] where 
\begin{align*}
	F_0(t)&=c_1e^{\lambda_1t}+c_2e^{\lambda_2t}\\
	F_1(t)&=\frac{1}{\nu_{01}}\left(c_1\left(d_0+\nu_{01}-\lambda_1\right)e^{\lambda_1t}+c_2\left(d_0+\nu_{01}-\lambda_2\right)e^{\lambda_2t}\right)
\end{align*}
\begin{align*}
	\lambda_{1,2}&=\frac{d_0+\nu_{01}+\nu_{10}\pm\sqrt{\left(d_0+\nu_{01}+\nu_{10}\right)^2-4d_0\nu_{10}}}{2},\\
	c_1&=\frac{\lambda_2-d_0}{\lambda_2-\lambda_1}e^{-\lambda_1T},\qquad \text{and} \qquad c_2=\frac{\lambda_1-d_0}{\lambda_1-\lambda_2}e^{-\lambda_2T}.
\end{align*}

Indifference pricing was first used in the pioneering paper of Hodges and Neuberger  \cite{HodgesNeu}. We refer also to  \cite{Carmona} for further applications (see \cite{Leung} and \cite{Zhou} as well).

The existence of classical solutions was proved  in \cite{LudShen}  when the payoff function $h(S)$ is bounded. This case is restrictive since it does not include such typical example  as the call option $h=\max\left\{S-K,0\right\}$ with strike price $K$. We investigate the solvability of the problem  and prove the existence and uniqueness of a weak solution in suitable Sobolev weighted spaces which allows unbounded terminal payoff functions. 

The integro-differential equation  \eqref{eq:psintdiff} is derived from \eqref{R01systemPDE} as follows. Denote $r^0:=\gamma R^0$, $r^1=\gamma R^1$. The system of differential equations for $r^0$ and $r^1$ has the following from:
\begin{equation}\label{DiffeqSyst}
	\left\{
	\begin{array}{l}\displaystyle
		r^0_\tau-\frac{1}{2}\sigma^2S^2r^0_{SS}=-\nu_{01}e^{-\left(r^1-r^0\right)}+d_0+\nu_{01}\\
		\displaystyle
		r^1_\tau=-\nu_{10}e^{-\left(r^0-r^1\right)}+\nu_{10}
	\end{array}
	\right.
\end{equation}
where $\tau=T-t$. The ODE in \eqref{DiffeqSyst} can be solved explicitly with respect to $r^1$. Then  we obtain the initial value problem \eqref{eq:psintdiff} under the substitution $u:=r^0-\nu_{10}\tau$ and $\kappa:=d_0+\nu_{01}-\nu_{10}$.

The paper is organized as follows. In Section 2 we prove a comparison principle (Theorem \ref{th:comppr}) for classical solutions to the problem \eqref{eq:psintdiff}. Then, in Section 3 we prove  a comparison principle (Theorem \ref{th:weakCompPr}) for weak sub/super solutions. In addition, we study the existence and uniqueness of weak solutions in a suitable weighted Sobolev space (see Theorem \ref{th:Exist1}).
\section{Comparison principle for classical solutions}

In this section we consider solutions of \eqref{eq:psintdiff} satisfying 
\begin{equation}	
	\left|u\right|,\left|h\right|\leq A\exp\left(\alpha\ln^2S\right)=AS^{\alpha\ln S}, \label{eq:growthcond}
\end{equation}
for some positive constants $A$ and $\alpha$. Note that conditions \eqref{eq:growthcond} include for example linear growth, polinomial and powers of $S$ with arbitrary exponent.

We prove the following comparison principle:
\begin{theorem} Let $u_1,u_0\in C(\left(0,+\infty\right)\times\left[0,T\right))\cap C^{2,1}(\left(0,+\infty\right)\times\left(0,T\right))$ be two clasical solutions of \eqref{eq:psintdiff} corresponding to the initial data $h=h_1$ and $h=h_0$, respectively and such that the conditions \eqref{eq:growthcond} hold. Then
\begin{equation}
	\gamma\inf\left(h_1-h_0\right)\leq u_1-u_0\leq \gamma\sup\left(h_1-h_0\right). \label{eq:comppr}
\end{equation}
\label{th:comppr}
\end{theorem}
We will only prove the lower bound in \eqref{eq:comppr} since the upper one follows immediately from it. In addition, we can assume that 
\[\underline{h}:=\gamma\inf\left(h_1-h_0\right)>-\infty,
\]
otherwise the left inequality in \eqref{eq:comppr} is trivial. We will use the following auxiliary lemma
\begin{lemma} Let $u_1$ and $u_0$ be as in Theorem \ref{th:comppr} and $\tau_1\geq 0$ be such that $u_1\left(S,\tau\right)-u_0\left(S,\tau\right)\geq \underline{h}$ for any $\tau\in[0,\tau_1]$. Then, there exists a constant $\bar{\tau}>0$ such that
$u_1\left(S,\tau\right)-u_0\left(S,\tau\right)\geq \underline{h}$ for any $\tau\in[0,\tau_1+\bar{\tau})$. In addition, $\bar{\tau}$ depends only on $\alpha$ defined in \eqref{eq:growthcond} and $\sigma$.\label{lm:aux}
\end{lemma} 
\begin{proof}
Let $u_1$ and $u_0$ be two solutions of \eqref{eq:psintdiff} corresponding to the initial conditions $u_1\left(S,0\right)=\gamma h_1(S)$ and $u_0\left(S,0\right)=\gamma h_0(S)$. Denote $\tilde{u}=u_1-u_0$, $\tilde{h}=\gamma\left(h_1-h_0\right)$, $u_\xi=\xi u_1+\left(1-\xi\right)u_0$, $h_\xi=\xi h_1 +\left(1-\xi\right)h_0$, for $\xi\in[0,1]$ and define
\[\mathcal{F}\left[\tau;u, g\right]:=-\nu_{01}e^{u(\tau )}\left(\nu_{10}\int_0^\tau e^{-u(s)}ds+e^{-g}\right)+\kappa.
\] 
Then
\begin{align}
	&\mathcal{F}\left[\tau;u_1, \gamma h_1\right]-\mathcal{F}\left[\tau;u_0,\gamma h_0\right]
	=\int_0^1\frac{d}{d\xi}\left(\mathcal{F}\left[\tau;u_\xi,\gamma h_\xi\right]\right) d\xi\\
	&\qquad=-\nu_{01}\tilde{u}\int_0^1e^{u_\xi(\tau )}\left(\nu_{10}\int_0^\tau e^{-u_\xi(s)}ds+e^{-\gamma h_\xi}\right)d\xi\\
	&\qquad\qquad +\nu_{01}\int_0^1e^{u_\xi(\tau )}\left(\nu_{10}\int_0^\tau e^{-u_\xi(s)}\tilde{u}\left(s\right)ds+e^{-\gamma h_\xi}\tilde{h}\right)d\xi\nonumber\\
	&\qquad=-\nu_{01}\nu_{10}\int_0^1\int_0^\tau e^{u_\xi(\tau )-u_\xi(s)}\left(\tilde{u}\left(\tau \right)-\tilde{u}\left(s\right)\right)dsd\xi\label{eq:difres}\\
	&\qquad\qquad -\nu_{01}\left(\tilde{u}\left(\tau \right)-\tilde{h}\right)\int_0^1e^{u_\xi(\tau )-\gamma h_\xi} d\xi\nonumber
\end{align}
and
\begin{align}\label{eq:psidiff}
	\tilde{u}_\tau-\frac{1}{2}\sigma^2S^2\tilde{u}_{SS}&=
	-\nu_{01}\nu_{10}\int_0^\tau \left(\tilde{u}\left(\tau \right)-\tilde{u}\left(s\right)\right)ds\int_0^1e^{u_\xi(\tau )-u_\xi(s)}d\xi\\
	&\qquad -\nu_{01}\left(\tilde{u}\left(\tau \right)-\tilde{h}\right)\int_0^1e^{u_\xi(\tau )-\gamma h_\xi} d\xi\nonumber
\end{align}

Next, define
\begin{equation}
	\omega(S,\tau):=\frac{1}{\sqrt{T_1-\tau}}\exp\left(\frac{\left(\ln S-\frac{1}{2}\sigma^2\left(T_1-\tau\right)\right)^2}{2\sigma^2\left(T_1-\tau\right)}\right), \label{eq:wdef}
\end{equation}
where  $T_1>0$  and $(S,\tau)\in\left(0,+\infty\right)\times\left[0,T_1\right)$. Note that $\mathcal{L}_{BS}\omega=\omega_\tau-\frac{1}{2}\sigma^2S^2\omega_{SS}=0$ and $\omega$ is increasing with respect to $\tau$ in the interval $\tau\in\left[T_1-4/\sigma^{2},T_1\right)$. 
Choose $T_1>\tau_1$ in \eqref{eq:wdef} such that the inequality
\[\alpha< \frac{1}{2\sigma^2\left(T_1-\tau\right)},
\]
holds for all $\tau\in[\tau_1,T_1)$ and $T_1-4/\sigma^{2}<\tau_1$. It is enough to define $T_1:=\tau_1+\bar{\tau}$, where $0<\bar{\tau}<\min\left\{\left(2\sigma^2\alpha\right)^{-1},4/\sigma^{2}\right\}$.
Next, let $\varphi_\epsilon=\tilde{u}+\epsilon \omega$.  Then
\begin{align}
	&\left(\varphi_\epsilon\right)_\tau-\frac{1}{2}\sigma^2S^2\left(\varphi_\epsilon\right)_{SS}=
	-\nu_{01}\nu_{10}\int_0^\tau \left(\tilde{u}\left(\tau \right)-\tilde{u}\left(s\right)\right)ds\int_0^1e^{u_\xi(\tau )-u_\xi(s)}d\xi\nonumber\\
	&\qquad -\nu_{01}\left(\tilde{u}\left(\tau \right)-\tilde{h}\right)\int_0^1e^{u_\xi(\tau )-\gamma h_\xi} d\xi\\
	&\qquad\geq-\nu_{01}\nu_{10}\left(\tilde{u}\left(\tau \right)-\underline{h}\right)\int_0^{\tau_1} ds\int_0^1e^{u_\xi(\tau )-u_\xi(s)}d\xi\label{eq:varphi}\\
	&\qquad\quad-\nu_{01}\nu_{10}\int_{\tau_1}^\tau \left(\tilde{u}\left(\tau \right)-\tilde{u}\left(s\right)\right)ds\int_0^1e^{u_\xi(\tau )-u_\xi(s)}d\xi\nonumber\\
	&\qquad -\nu_{01}\left(\tilde{u}\left(\tau \right)-\tilde{h}\right)\int_0^1e^{u_\xi(\tau )-\gamma h_\xi} d\xi\nonumber
\end{align}
We will prove that $\varphi_\epsilon\geq \underline{h}$ for any $\tau\in\left[\tau_1,T_1\right)$. Indeed, assume by contradiction that
$\inf\varphi_\epsilon< \underline{h}$. Note that $\left.\varphi_\epsilon\right|_{\tau=\tau_1}>\underline{h}$ and there exist $\bar{S}$ and $\underline{S}$ such that $\varphi_\epsilon>\underline{h}$ if either $S\leq \underline{S}$ or $S\geq \bar{S}$. In fact, $\varphi_\epsilon\rightarrow+\infty$ uniformly when either $\left|\ln S\right|\rightarrow+\infty$ or $\tau\rightarrow T_1$. The last observations imply that $\varphi_\epsilon$ attains minimum in an interior point $(S_*,\tau_*)\in (\underline{S},\bar{S})\times (\tau_1,T_1)$ and $\varphi_\epsilon(S_*,\tau_*)<\underline{h}$. Then, $\left(\varphi_\epsilon\right)_\tau(S_*,\tau_*)=0$, $\left(\varphi_\epsilon\right)_{SS}(S_*,\tau_*)\geq 0$ and 
\begin{align}
	\tilde{u}(S_*,\tau_*)-\tilde{h}&\leq \tilde{u}(S_*,\tau_*)-\underline{h}=\varphi_\epsilon(S_*,\tau_*)-\underline{h}-\epsilon \omega(S_*,\tau_*)<0\\
	\tilde{u}(S_*,\tau_*)-\tilde{u}(S_*,s)&=\varphi_\epsilon(S_*,\tau_*)-\varphi_\epsilon(S_*,s)\nonumber\\
	&\qquad-\epsilon\left(\omega(S_*,\tau_*)-\omega(S_*,s)\right)<0, \quad \forall s\in [\tau_1,\tau_*],
\end{align}
since $\omega$ is increasing in $\tau$. Thus the right hand side of \eqref{eq:varphi} is positive, a contradiction. 
Hence $\varphi_\epsilon=\tilde{u}+\epsilon \omega\geq\underline{h}$ for any $\tau\in\left[\tau_1,T_1\right)$. Let $\epsilon\rightarrow0$. Then $\tilde{u}=u_1-u_0\geq\underline{h}$ for any $\tau\in\left[\tau_1,T_1\right)$.
\end{proof}
\begin{proof} (of Theorem \ref{th:comppr}) The comparison principle follows by induction and the auxiliary Lemma \ref{lm:aux}: we first take $\tau_1=0$ and prove it in the interval $[0,1/2\bar{\tau}]$, then let $\tau_1=1/2\bar{\tau}$ and consider the interval $[1/2\bar{\tau},\bar{\tau}]$ and etc.
\end{proof}


%
Now, as a corollary we formulate  comparison principle for the buyer's indifference prices $p(S,t)$, $q(S,t)$ which satisfy  the terminal value problem 
\begin{equation}\label{eq:pqsyst}
	\left\{
	\begin{array}{l}\displaystyle
		p_t+\frac{1}{2}\sigma^2S^2p_{SS}-\frac{\nu_{01}}{\gamma}\frac{F_1}{F_0}e^{-\gamma\left(q-p\right)}+\frac{d_0+\nu_{01}}{\gamma}-\frac{1}{\gamma}\frac{F_0^\prime}{F_0}=0\\
		\displaystyle
		q_t-\frac{\nu_{10}}{\gamma}\frac{F_0}{F_1}e^{-\gamma\left(p-q\right)}+\frac{\nu_{10}}{\gamma}-\frac{1}{\gamma}\frac{F_1^\prime}{F_1}=0\\
		p(S,T)=q(S,T)=h(S).
	\end{array}
		\right.
\end{equation}
By  \textit{classical solutions} of \eqref{eq:pqsyst} we mean functions such that $p\in  C(\left(0,+\infty\right)\times\left(0,T\right])\cap C^{2,1}(\left(0,+\infty\right)\times\left(0,T\right))$, $q\in C(\left(0,+\infty\right)\times\left(0,T\right])$, $q_t\in C(\left(0,+\infty\right)\times\left(0,T\right))$. 

Note that  
\begin{align}
	\gamma p&=\nu_{10}(T-t)+\ln F_0(t)+u(S,T-t),\label{eq:ppsi}\\
	\gamma q&=\nu_{10}(T-t)+\ln F_1(t)-\ln\left(\nu_{10}\int_0^{T-t} e^{-u(S,s)}ds+e^{-\gamma h(S)}\right),\label{eq:qpsi}
\end{align}
since $p(t)=\gamma^{-1}\left(r^0+\ln F_0(t)\right)$ and $q(t)=\gamma^{-1}\left(r^1+\ln F_1(t)\right)$. Then, a comparison principle in $(p,q)$ solutions will be equivalent to a comparison principle for the $(r^0,r^1)$ variables.

We consider growth conditions analogous to \eqref{eq:growthcond}
\begin{equation}	
	\left|p\right|,\left|h\right|\leq A\exp\left(\alpha\ln^2S\right)=AS^{\alpha\ln S}, \label{eq:growthcondp}
\end{equation}
for some positive constants $A$ and $\alpha$.
\begin{corollary} Let $\left(p_1,q_1\right)$ and $\left(p_0,q_0\right)$ be two classical solutions of the system \eqref{eq:pqsyst} corresponding to terminal data $h\equiv h_1(S)$ and $h\equiv h_0(S)$, respectively. If there exist some positive constants $A$ and $\alpha$ such that $p_i(S,t)$ and $h_i(S)$, $i=0,1$ satisfy the conditions \eqref{eq:growthcondp}, then
\begin{gather}
	\inf\left(h_1-h_0\right)\leq p_1(S,t)-p_0(S,t)\leq \sup\left(h_1-h_0\right),\label{eq:CPp}\\
	\inf\left(h_1-h_0\right)\leq q_1(S,t)-q_0(S,t)\leq \sup\left(h_1-h_0\right).\label{eq:CPq}
\end{gather}

In particular, let $h(S)$ be bounded from below (or from above) by a constant, i.e. 
 $h(S)\geq h_*$ (resp. $h(S)\leq h^*$) and $p(S,t)$, $q(S,t)$, be a classical solutions of the terminal value problem \eqref{eq:pqsyst} satisfying \eqref{eq:growthcondp}. Then 
\[p(S,t)\geq h_* \;\text{and}\;q(S,t)\geq h_* \ \text{(respectively } p(S,t)\leq h^* \;\text{and}\;q(S,t)\leq h^*\text{)},
\]
for any $S\in (0,+\infty)$ and any $t\in (0,T]$.
\end{corollary}
\begin{proof}
The inequalities \eqref{eq:CPp} follow immediately from Theorem \ref{th:comppr} and representation \eqref{eq:ppsi}. In order to prove \eqref{eq:CPq} we will use \eqref{eq:qpsi}, i.e.
\[q_i(\cdot,t)=\gamma^{-1}\left[\nu_{10}(T-t)+\ln F_1(t)-\ln\left(\nu_{10}\int_0^{T-t} e^{-u_i(\cdot,s)}ds+e^{-\gamma h_i(\cdot)}\right)\right],
\]
for $i=0,1$.
Similarly to the proof of Lemma \ref{lm:aux} we derive
\begin{align*}
	q_1(\cdot,t)-q_0(\cdot,t)&=-\gamma^{-1}\int_0^1\frac{d}{d\xi}\left[\ln\left(\nu_{10}\int_0^{T-t} e^{-u_\xi(\cdot,s)}ds+e^{-\gamma h_\xi(\cdot)}\right)\right]d\xi\\
	&=\gamma^{-1}\int_0^1\frac{\nu_{10}\int_0^{T-t} e^{-u_\xi(\cdot,s)}\left(u_1(\cdot,s)-u_0(\cdot,s)\right)ds}{\nu_{10}\int_0^{T-t} e^{-u_\xi(\cdot,s)}ds+e^{-\gamma h_\xi(\cdot)}}d\xi\\
	&\qquad\qquad+(h_1(\cdot)-h_0(\cdot))\int_0^1\frac{e^{-\gamma h_\xi(\cdot)}}{\nu_{10}\int_0^{T-t} e^{-u_\xi(\cdot,s)}ds+e^{-\gamma h_\xi(\cdot)}}d\xi
\end{align*}
Now, \eqref{eq:comppr} implies the estimates \eqref{eq:CPq}.

The second part follows immediately due to the fact that $p_*(S,t)\equiv h_*$ and $q_*(S,t)\equiv h_*$ are the solutions of the problem \eqref{eq:pqsyst} with constant terminal condition $h\equiv h_*$. Indeed, if we formally substitute $p_*(S,t)\equiv h_*$ and $q_*(S,t)\equiv h_*$ in \eqref{eq:pqsyst}, then we arrive at the conclusion that it is sufficient to check the following identities
\begin{align}
	-\frac{\nu_{01}}{\gamma}\frac{F_1}{F_0}+\frac{d_0+\nu_{01}}{\gamma}-\frac{1}{\gamma}\frac{F_0^\prime}{F_0}=0,\\
	-\frac{\nu_{10}}{\gamma}\frac{F_0}{F_1}+\frac{\nu_{10}}{\gamma}-\frac{1}{\gamma}\frac{F_1^\prime}{F_1}=0,
\end{align}
or equivalently
\begin{align}
	F_0^\prime&=-\nu_{01}F_1+\left(d_0+\nu_{01}\right)F_0,\\
	F_1^\prime&=-\nu_{10}F_0+\nu_{10}F_1,
\end{align}
which follow directly from the definition of $F_0$ and $F_1$.
\end{proof}

\section{Existence of weak solutions}

In this section we study the existence and uniqueness of weak solutions in suitable function spaces. First we introduce the weighted $L^2$ space 
\[L_w^2:=\left\{u:\left\|u\right\|^2_{0}:=\int_0^{+\infty}u^2(S)w(S)dS<\infty\right\},
\]
given a weight function $w>0$. Then we define a weighted Sobolev space as follows
\[H_w^1:=\left\{u:u\in L_w^2  \text{ s.t. } Su^\prime(S) \in L_w^2\right\},
\]
with norm $\left\|\cdot\right\|_{1}$ such that $\left\|u\right\|^2_{1}=\left\|u\right\|^2_{0}+\left\|Su^\prime\right\|^2_{0}$.

Let $\xi:[0,+\infty)\rightarrow[0,1]$  be increasing, infinitely continuously differentiable function and such that $\xi\equiv 0$ on $[0,1/2]$ and $\xi\equiv 1$ on $[1,+\infty)$. We will use $\xi$ to construct a sequence $\left\{u_\epsilon\right\}$ of compactly supported functions converging in $H^1_w$  to a given element $u\in H^1_w$. More precisely, the following auxiliary result holds. 
\begin{lemma} \label{lm:xieps} Let $\xi_\epsilon(x):=\xi(x/\epsilon)\left[1-\xi(x\epsilon/2)\right]$, $0<\epsilon<1$ and $u_\epsilon:=\xi_\epsilon u$. Then $u_\epsilon\rightarrow u$ in $H^1_w$, as $\epsilon\rightarrow 0$.
\end{lemma}
\begin{proof} Note that $\left(u-u_\epsilon\right)^\prime=\left(1-\xi_\epsilon\right)u^\prime-\xi_\epsilon^\prime u$, 
\[S\xi_\epsilon^\prime(S)=\left(S/\epsilon\right)\xi^\prime(S/\epsilon)\left[1-\xi(S\epsilon/2)\right]-\left(S\epsilon/2\right)\xi^\prime(S\epsilon/2)\xi(S/\epsilon)\] 
is uniformly bounded with respect to $\epsilon$ and $1-\xi_\epsilon\rightarrow 0$ as well as $S\xi_\epsilon^\prime(S) \rightarrow 0$ as $\epsilon\rightarrow0$. Then the Lebesgue's dominated convergence theorem implies that $\left\|u-u_\epsilon\right\|\rightarrow 0$ as $\epsilon\rightarrow0$.
\end{proof}

Next, let $u(S)$ be twice continuously differentiable on $(0,+\infty)$ and denote the operator $\mathcal{L}u:=-\frac{1}{2}\sigma^2S^2u^{\prime\prime}$. Then after integration by parts we formally obtain:
\begin{align*}
	(\mathcal{L}u,v)_{L_w^2}&=-\frac{1}{2}\sigma^2\int_0^{+\infty}wS^2u^{\prime\prime}vdS\\
	&=\frac{1}{2}\sigma^2\int_0^{+\infty}\left[wS^2u^{\prime}v^{\prime}+\left(S\frac{w^\prime}{w}+2\right)wSu^{\prime}v\right]dS,
\end{align*}
provided that the integrals above are well-defined, $w$ is continuously differentiable and $wS^2u^{\prime}v\rightarrow$ as $S\rightarrow0$ and $S\rightarrow\infty$. For example, the above holds when $v$ is continuously differentiable and with compact support.

Following the above observations we introduce the bilinear form:
\begin{equation}\label{eq:abilindef}
	a(u,v):=\frac{1}{2}\sigma^2\int_0^{+\infty}wSu^{\prime}\left[Sv^{\prime}+\left(S\frac{w^\prime}{w}+2\right)v\right]dS.
\end{equation}
If the weight function $w$ 
is twice continuously differentiable, and there exists a constant $C>0$, such that 
\begin{equation}
\left|S\frac{w^\prime(S)}{w(S)}\right|,\left|S^2\frac{w^{\prime\prime}(S)}{w(S)}\right|\leq C, \forall S\in (0,+\infty).
\label{eq:wassump}
\end{equation} 
then the bilinear form $a(u,v)$ is continuous and semi-coercive on $H_w^1$, i.e., 
\begin{align}
	\left|a\left(u,v\right)\right|&\leq c\left\|u\right\|_1\left\|v\right\|_1, \quad \forall u,v\in H_w^1\label{eq:acont}\\
	a\left(u,u\right)&\geq \alpha\left\|u\right\|_1^2-\beta\left\|u\right\|_0^2, \quad \forall u\in H_w^1\label{eq:semicoerc}
\end{align}
for some suitable constants $c>0$, $\alpha>0$ and $\beta>0$ which are independent of $u$ and $v$.

We can choose such weight function that the call option payoff function $h=\max\left\{S-K,0\right\}$  belongs to the space $H_w^1$, for example, take $w:=(1+S)^\gamma$, where $\gamma<-3$. 

In addition, we assume that 
\begin{equation}
\theta:=\int_0^{+\infty} w(S)dS<+\infty.
\label{eq:wassumpintegr}
\end{equation}
This assumption guarantees that any bounded and measurable function belongs to $L^2_w$.
\begin{lemma} \label{lm:Est}There exists  a constant $c_0>0$ such that 
\begin{equation}
\left|u(S)\right|^2\leq c_0\left\|u\right\|_1^2\frac{1}{S} \exp (C\left|\ln S\right|) ,\qquad \forall u\in H^1_w,
\label{eq:Est1}
\end{equation}
where $C$ satisfies \eqref{eq:wassump}.
\end{lemma}
\begin{proof}
Note that there exists a constant $c_0$ such that
\begin{equation}
\left|u(1)\right|^2\leq c_0\left\|u\right\|_1^2,\qquad \forall u\in H^1_w,
\label{eq:Estu1}
\end{equation}
due to the Sobolev embbeding theorem.

Let $S$ be fixed and denote $v(\zeta):=u(\zeta S)$. We have
\begin{align}
	\left\|v\right\|_1^2&=\int_0^{+\infty}w(\zeta)\left(\zeta^2S^2\left(u^\prime(\zeta S)\right)^2+u^2(\zeta S)\right)d\zeta\\
	&=\int_0^{+\infty}\frac{w(\zeta)}{Sw(\zeta S)}w(\zeta S)\left(\zeta^2S^2\left(u^\prime(\zeta S)\right)^2+u^2(\zeta S)\right)d\left(S\zeta\right)\\
	&\leq \frac{1}{S} \exp (C\left|\ln S\right|) \left\|u\right\|_1^2,
\end{align}
since
\[\frac{w(\zeta)}{Sw(\zeta S)}=\frac{1}{S} \exp \left(\int_{\zeta S}^\zeta \frac{w^\prime\left(\xi\right)}{w\left(\xi\right)}d\xi\right)\leq \frac{1}{S} \exp \left(C\left|\ln S\right|\right).
\]
Then \eqref{eq:Est1} follows from \eqref{eq:Estu1} since $v(1)=u(S)$.
\end{proof}

The space $H^1_w$ is densely and continuously embbeded in $L^2_w$. We consider the Gelfand triples
\[H^1_w\subset L^2_w\subset H^{*}_w,
\]
and 
\[L^2(0,T;H^1_w)\subset L^2(0,T;L^2_w)\subset L^2(0,T;H^{*}_w),
\]
where $H^{*}_w$ is the dual of $H^1_w$. Next, we define the set
\begin{equation}
W(0,T):=\left\{u\in L^2(0,T;H^1_w),\dot{u}\in L^2(0,T;H^*_w)\right\},
\label{eq:Wdef}
\end{equation}
where $\dot{u}$ is the distributional derivative of $u$. It is well known (see Lions and Magenes\cite{LionsMag}) that
\[W(0,T)\subset C([0,T],L^2_w).
\]

For simplicity we will further write $u(\tau)$ instead of $u(S,\tau)$ when this does not lead to misunderstanding.
Recall that 
\[\mathcal{F}\left[\tau;u,\gamma h\right]:=-\nu_{01}e^{u(\tau )}\left(\nu_{10}\int_0^\tau e^{-u(s)}ds+e^{-\gamma h}\right)+\kappa.
\] 
\begin{definition} \label{def:Soldef} A function $u\in W(0,T)$ is called \textit{weak supersolution (subsolution)} of the initial value problem \eqref{eq:psintdiff} if $u(0)\geq \gamma h$ (resp. $u(0)\leq \gamma h $) and for a.a. $\tau\in (0,T)$ the inequality
\begin{equation}
\left\langle\dot{u},v\right\rangle+a(u,v)\geq (\leq) \int_0^{+\infty}w\mathcal{F}\left[\tau;u,\gamma h\right]vdS,
\label{eq:SupSoldef}
\end{equation}
holds for any nonegative $v\in H^1_w$.
Respectively, the function $u\in W(0,T)$ is called \textit{weak solution} of the initial value problem \eqref{eq:psintdiff} if $u(0)=\gamma h$ and for a.a. $\tau\in (0,T)$ the equality 
\begin{equation}
\left\langle\dot{u},v\right\rangle+a(u,v)=\int_0^{+\infty}w\mathcal{F}\left[\tau;u,\gamma h\right]vdS, \qquad \forall v\in H^1_w,
\label{eq:WeakSoldef}
\end{equation}
holds.
\end{definition}
Next, we prove the following comparison principle for  
weak super/subsolutions satisfying growth conditions of type \eqref{eq:growthcond}.
\begin{theorem} \label{th:weakCompPr} Let $\overline{u}$ be a weak supersolution of the initial value problem \eqref{eq:psintdiff} with initial data $h(S)\equiv \overline{h}$ and $\underline{u}$ be a weak subsolution corresponding to the initial data $h(S)\equiv \underline{h}$ where $\underline{h}$ and $\overline{h}$ are given and $\underline{h}\leq \overline{h}$. Assume in addition, that  there exist positive constants $A$ and $\alpha$ such that 
\begin{equation}
\left|\underline{h}\right|,\left|\overline{h}\right|,\left|\underline{u}\right|,\left|\overline{u}\right|\leq A\exp\left(\alpha\ln^2S\right)=AS^{\alpha\ln S}, 
\label{eq:subsupgrcnd}
\end{equation}
 for a.a. $(S,t)\in(0,+\infty)\times [0,T]$.

Then $\underline{u}\leq\overline{u} $ for a.a. $(S,t)\in(0,+\infty)\times [0,T]$.
\end{theorem}
Denote $u:=\overline{u}-\underline{u}$. We will prove that $u_-:=\max\left\{-u,0\right\}=0$ almost everywhere. 
Similarly to \eqref{eq:psidiff},  we obtain that the following inequality holds for a.a. $\tau\in(0,T)$ and for any nonegative $v\in H^1_w$ with compact support in $(0,+\infty)$:
\begin{align}\label{eq:udiffcomp}
\left\langle\dot{u},v\right\rangle+a(u,v)&\geq -\nu_{01}\nu_{10}\int_0^\infty \left(\int_0^\tau \delta(\tau,s)\left(u\left(S,\tau \right)-u\left(S,s\right)\right)ds \right)v(S)wdS\\
	&\qquad -\nu_{01}\int_0^\infty \left(u\left(S,\tau \right)-\tilde{h}(S)\right) v(S)\delta(\tau)wdS\nonumber,
\end{align}
where 
\[\delta(\tau,s):=\int_0^1e^{u_\xi(\tau )-u_\xi(s)}d\xi, \quad \delta(\tau):=\int_0^1e^{u_\xi(\tau )-\gamma h_\xi} d\xi,
\]
$u_\xi:=\xi \overline{u}+\left(1-\xi\right)\underline{u}$, $u(\cdot,0)\geq \tilde{h}:=\gamma\left(\overline{h}-\underline{h}\right)\geq 0$ and $h_\xi:=\xi \overline{h}+\left(1-\xi\right)\underline{h}$.

It is sufficient to prove the following auxiliary result:
\begin{lemma} Assume that  $\tau_1\geq0$ is such that for any $t\in[0,\tau_1]$ the inequality $\overline{u}(t)-\underline{u}(t)\geq 0$ holds a.e. on $(0,+\infty)$. Then the same inequality holds for any $t\in[0,\tau_1+\bar{\tau}]$, 
 where $\bar{\tau}>0$ is a constant which depends only on $\alpha$ and $\sigma$. 
\end{lemma}
\begin{proof} Let $\omega$ be defined by \eqref{eq:wdef} and $u_\epsilon:=u+\epsilon\omega$ where $u=\overline{u}-\underline{u}$. Then, assume that $\bar{\tau}$ is chosen as in the proof of Lemma \ref{lm:aux}. We will prove that $u_{\epsilon-}:=\max\left\{-u_\epsilon,0\right\}\equiv 0$ for a.a. $(S,t)\in(0,+\infty)\times [\tau_1,\tau_1+\bar{\tau}]$. Note that  there exist a closed interval $I_\epsilon\subset(0,+\infty)$ such that $u_{\epsilon-}=0$ on the set $\left((0,+\infty)\setminus I_\epsilon\right)\times [\tau_1,\tau_1+\bar{\tau}]$ due to the conditions \eqref{eq:subsupgrcnd}. Now, let $\varphi(S)$ be a smooth function with compact support in $(0,+\infty)$ such that $\varphi(S)=1$ on the interval $I_\epsilon$. Then $u_\epsilon \varphi\in L^2(\tau_1,\tau_1+\bar{\tau};H^1_w)$ and $\left(u_\epsilon \varphi\right)_-=u_{\epsilon-}$. Next, for any nonnegative $v\in H^1_w$ with compact support $\supp v\subset I_\epsilon$ we have $\varphi v=v$, $a(u\varphi,v)=a(u,v)$ and then
\begin{align}
&\left\langle \frac{d}{d\tau}\left(u_\epsilon \varphi\right),v\right\rangle+a\left(u_\epsilon \varphi,v\right)=\left\langle\dot{u},\varphi v\right\rangle+\epsilon\left\langle \varphi\dot{\omega},v\right\rangle +a(u\varphi,v)+\epsilon a(\omega\varphi,v)\label{eq:ueps}\\
&\quad=\left\langle\dot{u},v\right\rangle+a(u,v)-\frac{1}{2}\epsilon \sigma^2\underbrace{(2\omega^\prime\varphi^\prime+\omega\varphi^{\prime\prime},v)_{L^2_w}}_{=0}\nonumber\\
&\quad\geq -\nu_{01}\nu_{10}\int_0^\infty \left(\int_0^\tau \delta(\tau,s)\left(u\left(S,\tau \right)-u\left(S,s\right)\right)ds \right)v(S)wdS\\
	&\qquad -\nu_{01}\int_0^\infty \left(u\left(S,\tau \right)-\tilde{h}(S)\right) v(S)\delta(\tau)wdS\nonumber\\
	&\quad\geq -\nu_{01}\nu_{10}\int_0^\infty \left(\int_0^{\tau_1} \delta(\tau,s)ds \right)u\left(S,\tau \right)v(S)wdS\\
	&\quad\quad -\nu_{01}\nu_{10}\int_0^\infty \left(\int_{\tau_1}^\tau \delta(\tau,s)\left(u\left(S,\tau \right)-u\left(S,s\right)\right)ds \right)v(S)wdS\nonumber\\
	&\qquad -\nu_{01}\int_0^\infty u\left(S,\tau \right) v(S)\delta(\tau)wdS,\nonumber
\end{align}
i.e.,
\begin{align}
&\left\langle \frac{d}{d\tau}\left(u_\epsilon \varphi\right),v\right\rangle+a\left(u_\epsilon \varphi,v\right)\geq -\nu_{01}\nu_{10}\int_0^\infty \left(\int_0^{\tau_1} \delta(\tau,s)ds \right)u_\epsilon\left(S,\tau \right)v(S)wdS\label{eq:ueineq}\\
	&\quad\quad -\nu_{01}\nu_{10}\int_0^\infty \left(\int_{\tau_1}^\tau \delta(\tau,s)\left(u_\epsilon\left(S,\tau \right)-u_\epsilon\left(S,s\right)\right)ds \right)v(S)wdS\nonumber\\
	&\qquad -\nu_{01}\int_0^\infty u_\epsilon\left(S,\tau \right) v(S)\delta(\tau)wdS,\nonumber
\end{align}
where we have used the fact that $u_\epsilon>u$ and $u_\epsilon\left(S,\tau \right)-u_\epsilon\left(S,s\right)>u\left(S,\tau \right)-u\left(S,s\right)$ for any $s\in[\tau_1,\tau]$ since $\omega(S,\cdot)$ is increasing on that interval. Now, take $v=u_{\epsilon-}$ and note that $u_\epsilon=u_{\epsilon+}-u_{\epsilon-}$,  $a\left(u_\epsilon \varphi,u_{\epsilon-}\right)=-a\left(u_{\epsilon-} ,u_{\epsilon-}\right)$ and 
\[u_\epsilon\left(S,s\right)u_{\epsilon-}\left(S,\tau\right)\geq -u_{\epsilon-}\left(S,s\right)u_{\epsilon-}\left(S,\tau\right)\geq -\frac{1}{2}\left(u_{\epsilon-}^2\left(S,s\right)+u_{\epsilon-}^2\left(S,\tau\right)\right). 
\]
After integration with respect to $\tau$ form $\tau_1$ to $t\in[\tau_1,\tau_1+\bar{\tau}]$ the inequality \eqref{eq:ueineq} implies
\begin{equation}\label{eq:preestimate}
\frac{1}{2}\left\|u_{\epsilon-}(t)\right\|_0^2+a\left(u_{\epsilon-} ,u_{\epsilon-}\right)\leq -\int_{\tau_1}^t\left(\int_0^\infty \Sigma(S,\tau)u_{\epsilon-}^2\left(S,\tau \right)wdS\right)d\tau,
\end{equation}
where 
\[\Sigma(S,\tau):=\nu_{01}\nu_{10}\left(\int_0^{\tau_1} \delta(\tau,s)ds +\frac{1}{2}\int_{\tau_1}^\tau\delta(\tau,s)ds-\frac{1}{2}\int_{\tau}^t \delta(s,\tau)ds\right)+\nu_{01}\delta(\tau).
\]
$\left|\Sigma(S,\tau)\right|$ is bounded from above by a constant, say $C>0$, when $S\in I_\epsilon$  and due to the semi-coercivity of the bilinear form $a(\cdot,\cdot)$ (see \eqref{eq:semicoerc}) we obtain:
\begin{equation}
\frac{1}{2}\left\|u_{\epsilon-}(t)\right\|_0^2\leq \left(C+\beta\right)\int_{\tau_1}^t\left\|u_{\epsilon-}(\tau)\right\|_0^2d\tau.
\label{eq:GronwallEstimate}
\end{equation}
Hence the Gronwall inequality implies $\left\|u_{\epsilon-}(t)\right\|_0=0$ for any $t\in[\tau_1,\tau_1+\bar{\tau}]$ since $\left\|u_{\epsilon-}(\tau_1)\right\|_0=0$. Then $u+\epsilon\omega\geq 0$ a.e. Thus $u\geq 0$ a.e. since $\epsilon>0$ is arbitrary. 
\end{proof} 
We further prove another useful estimate.
\begin{lemma}  There exists a constant $C>0$ such that 
\begin{equation}
\max_{t\in[0,T]} \left\|u(t)\right\|_0+\left\|u\right\|_{L^2(0,T,H^1_w)}\leq C\left(\left\|u(0)\right\|_0+\left\|\hat{u}\right\|_{W(0,T)}+\gamma\left\|h\right\|_0+1\right)
\label{eq:Aprioriest1}
\end{equation} 
for any  weak subsolution $u$ and any function $\hat{u}\in W(0,T)$ satisfying $u\geq \hat{u}$.
\end{lemma}
\begin{proof} Let $v\in H^1_w$ be some nonnegative function. We have
\begin{align}
\left\langle\dot{u},v\right\rangle+a(u,v)&\leq \int_0^{+\infty}w\mathcal{F}\left[\tau;u,\gamma h\right]vdS,\nonumber
\\
&\leq -\nu_{01}\nu_{10}\left(\int_0^\tau \left[u(\tau )-u(s)\right]ds\ ,\ v\right)_{L^2_w}\label{eq:EstAbv}\\
&\quad-\nu_{01}\left(u(\tau )-\gamma h\ ,\ v\right)_{L^2_w}\nonumber\\
&\quad+\left(\kappa-\nu_{01}\nu_{10}\tau-\nu_{01}\right)\left(1\ ,\ v\right)_{L^2_w}.\nonumber
\end{align}
Take $v=u-\hat{u}$ and integrate \eqref{eq:EstAbv} with respect to $\tau$ from $0$ to $t$. 
\begin{align}
\frac{1}{2}\left\|u(t)\right\|_0^2+a(u,u)&\leq \frac{1}{2}\left\|u(0)\right\|_0^2+\left.\left(u,\hat{u}\right)_{L^2_w}\right|_0^t+a(u,\hat{u})-\int_0^t\left\langle\dot{\hat{u}},u\right\rangle d\tau\\
&\quad-\nu_{01}\int_0^t\left(\nu_{10}\tau+1\right)\left\|u(\tau)\right\|_0^2d\tau+\nu_{01}\nu_{10}\frac{1}{2}\left\|\int_0^tu(\tau )d\tau\right\|_0^2\nonumber
\\
&\quad+C_1\left(\left\|\hat{u}\right\|_{L^2(0,t,L^2_w)}+\gamma\left\|h\right\|_0+1\right)\left\|u\right\|_{L^2(0,t,L^2_w)}\nonumber\\
&\quad+C_2\left(\gamma\left\|h\right\|_0+1\right)\left\|\hat{u}\right\|_{L^2(0,t,L^2_w)}.\nonumber
\end{align}
Then a standard argument implies the estimate \eqref{eq:Aprioriest1}.
\end{proof}

Now, we prove the existence of weak solutions, provided that $h\in H^1_w$. The proof is based on the lower and upper solution method (cf. \cite{Pao}). However, the exponential nonlinearity in \eqref{eq:psintdiff} causes some very technical difficulties which have to be overcome.  
\begin{theorem} \label{th:Exist1} Assume that $h\in H^1_w$. Then there exist a
weak solution $u$ to the initial value problem \eqref{eq:psintdiff}. Moreover, there exists a constant $C>0$ independent of $u$ such that 
\begin{equation}
\left\|\dot{u}\right\|_{L^2(0,T,L^2_w)}+\left\|u\right\|_{L^\infty(0,T,H^1_w)}\leq C\left(\left\|u(0)\right\|_1+1\right)
\label{eq:Aprioriest2}
\end{equation} 
\end{theorem}
\begin{proof}
We will present the proof in several steps.

\noindent \textbf{Step 1.} \textit{Let $h\in L^2_w$ be bounded. Then there exists a weak solution $u$ to the initial value problem \eqref{eq:psintdiff}. In addition, if $u(0)=\gamma h\in H^1_w$, then the inequality \eqref{eq:Aprioriest2} holds with a constant $C$ independent of $u(0)$.} 

Note that we can conctruct appropriate couple of a supersolution $\overline{u}$ and a subsolution $\underline{u}$. Indeed, let the constant $c_0$ be such that $\left|\gamma h\right|\leq c_0$ and take $\underline{u}:=-c_0-Mt$ for some positive constant $M$. If $M$ is great enough then $\underline{u}$ is a subsolution. Analogously, $\overline{u}:=c_0+Mt$ is a supersolution provided that $M\geq \kappa$. Next, according to \eqref{eq:difres} we can choose a constant $N>0$ such that 
\[Nu(\tau)+\mathcal{F}\left[\tau;u,\gamma h\right]=Nu(\tau)-\nu_{01}e^{u(\tau )}\left(\nu_{10}\int_0^\tau e^{-u(s)}ds+e^{-\gamma h}\right)+\kappa
\] 
is increasing in $u$, i.e. \[Nu_1(\tau)+\mathcal{F}\left[\tau;u_1,\gamma h\right]\geq Nu_0(\tau)+\mathcal{F}\left[\tau;u_0,\gamma h\right],\]
for all $u_0$ and $u_1$ such that $\underline{u}\leq u_0\leq u_1 \leq \overline{u}$. Now, we can  construct a decreasing sequence of supersolutions $u_0:=\overline{u}, u_1, u_2, ...$ such that $u_{n+1}$ is the solution of the initial value problem
\[\left\{\begin{array}{l}
\dot{u}_{n+1}-\frac{1}{2}\sigma^2S^2u^{\prime\prime}_{n+1,SS}+Nu_{n+1}=Nu_n+\mathcal{F}\left[\tau;u_n,\gamma h\right],\\	
	u_{n+1}(S,0)=\gamma h(S)
\end{array}\right.
\]
and $\underline{u}\leq u_n \leq \overline{u}$. A standard argument implies that $u_n$ converges to a weak solution of the problem \eqref{eq:psintdiff}. We omit the details.

Next, assume in addition that $h\in H^1_w$. Then $\dot{u}\in L^2(0,T;L^2_w)$ and $u\in L^\infty(0,T;H^1_w)$ (see, e.g., Bonnans \cite{Bon}) and the following parabolic estimate holds:
\[
\left\|\dot{u}\right\|_{L^2(0,T,L^2_w)}+\left\|u\right\|_{L^\infty(0,T,H^1_w)}\leq c_0\left(\left\|u(0)\right\|_1+\left\|\mathcal{F}\left[\cdot;u,\gamma h\right]\right\|_{L^2(0,T,L^2_w)}\right)
\]
We will prove the stronger estimate \eqref{eq:Aprioriest2}. First, we have 
\begin{align}
	-\frac{1}{2}\sigma^2S^2u^{\prime\prime}_{SS}&=\mathcal{F}\left[\tau;u,\gamma h\right]-\dot{u}\in L^2(0,T,L^2_w),\\
	-\frac{1}{2}\sigma^2\int_0^t\left(S^2u^{\prime\prime}_{SS},\dot{u}\right)_{L^2_w}d\tau&=\frac{1}{2}\sigma^2\left(\frac{1}{2}\left\|u(t)\right\|_1^2-\frac{1}{2}\left\|u(0)\right\|_1^2\right)\label{eq:ussudot}\\
	&\quad+\frac{1}{2}\sigma^2\int_0^t\left(S\left(S\frac{w^\prime}{w}+2\right)u^\prime_S-u,\dot{u}\right)_{L^2_w}d\tau\nonumber
	\end{align}
	\begin{align}
	\int_0^t\left(\mathcal{F}\left[\tau;u,\gamma h\right],\dot{u}\right)_{L^2_w}d\tau&=\int_0^{+\infty}\left(\int_0^t\frac{d}{d\tau}\left(\mathcal{F}\left[\tau;u,\gamma h\right]\right)d\tau\right)\:wdS\\
	&\quad+\int_0^{+\infty}\left(\int_0^t\left(\kappa  \dot{u} + \nu_{01}\nu_{10}\right)d\tau  \right)\:wdS \nonumber\\
	&\leq \left|\kappa\right|\theta^{1/2}\int_0^t\left\|\dot{u}\left(\tau\right)\right\|_0d\tau+\nu_{01}\left(1+\nu_{10}t\right)\theta\label{eq:Fudot}
\end{align}
since
\begin{align}
	\frac{d}{d\tau}\left(\mathcal{F}\left[\tau;u,\gamma h\right]\right)&=\frac{d}{d\tau}\left[-\nu_{01}e^{u(\tau )}\left(\nu_{10}\int_0^\tau e^{-u(s)}ds+e^{-\gamma h}\right)+\kappa\right]\nonumber\\
	&=-\nu_{01}e^{u(\tau )}\left(\nu_{10}\int_0^\tau e^{-u(s)}ds+e^{-\gamma h}\right)\dot{u}-\nu_{01}\nu_{10}\\
	&=\mathcal{F}\left[\tau;u,\gamma h\right]\dot{u}-\kappa\dot{u}-\nu_{01}\nu_{10}.
\end{align}
and
\begin{equation}
\int_0^t\frac{d}{d\tau}\left(\mathcal{F}\left[\tau;u,\gamma h\right]\right)d\tau=\mathcal{F}\left[t;u,\gamma h\right]-\mathcal{F}\left[0;u,\gamma h\right]\leq \nu_{01}
\label{eq:intfot}
\end{equation} 
We multiply  both sides of the equation $\dot{u}-1/2\sigma^2S^2u^{\prime\prime}_{SS}=\mathcal{F}\left[\tau;u,\gamma h\right]$ with $\dot{u}$ in $L^2_w$ and integrate from $0$ to $T$. Then \eqref{eq:ussudot} and \eqref{eq:Fudot} imply
\begin{align}
	\int_0^t\left\|\dot{u}\right\|_0^2d\tau+\frac{1}{4}\sigma^2\left\|u(t)\right\|_1^2&\leq-\frac{1}{2}\sigma^2\int_0^t\left(S\left(S\frac{w^\prime}{w}+2\right)u^\prime_S-u,\dot{u}\right)_{L^2_w}d\tau\\
	&\qquad+\left|\kappa\right|\theta^{1/2}\int_0^t\left\|\dot{u}\left(\tau\right)\right\|_0d\tau+\frac{1}{4}\sigma^2\left\|u(0)\right\|_1^2\nonumber\\
&\qquad\;+\nu_{01}\left(1+\nu_{10}t\right)\theta\nonumber\\
&\leq \tilde{C}\left[\int_0^t\left(\left\|u(\tau)\right\|_1+1\right)\left\|\dot{u}(\tau)\right\|_0d\tau+\left\|u(0)\right\|_1^2+1\right]\nonumber
\end{align}
for some constant $\tilde{C}>0$. Now, a techical, but standard argument implies that \eqref{eq:Aprioriest2} holds.

\noindent \textbf{Step 2.} \textit{Let $h\in H^1_w$ be bounded from below, i.e., $u(0)=\gamma h\geq c$. Then there exists a weak solution $u$ to the initial value problem \eqref{eq:psintdiff}. In addition, the inequality \eqref{eq:Aprioriest2} holds.} 

Let $\xi_\epsilon(x)$ be defined as in Lemma \ref{lm:xieps}, i.e., $\xi_\epsilon(x):=\xi(x/\epsilon)\left[1-\xi(x\epsilon/2)\right]$. Step 1 implies that there exists a solution $u_\epsilon$ corresponding to the initial condition $u_\epsilon(0)=\xi_\epsilon (\gamma h-c)+c=\xi_\epsilon \gamma h+(1-\xi_\epsilon )c$ which is bounded. Moreover, $\xi_\epsilon \gamma h+(1-\xi_\epsilon )c\leq\gamma h$ increases as $\epsilon\downarrow 0$  and converges in $H^1_w$ to $\gamma h$. Then the comparison principle from Theorem \ref{th:weakCompPr} implies that the sequence $u_\epsilon$ is increasing as $\epsilon\downarrow 0$. 
Next, the estimate \eqref{eq:Aprioriest2} and Lemma \ref{lm:Est} imply that $u_\epsilon(S,\tau)$ converges to a finite limit $u(S,\tau)$ for any $(S,\tau)\in \left(0,+\infty\right)\times[0,T]$. What is more, $\dot{u}_\epsilon$ is weakly convergent to $\dot{u}(S,\tau)$ in $L^2(0,T;L^2_w)$, $u_\epsilon$ is weakly-$*$ convergent to $u$ in $L^\infty(0,T,H^1_w)$ and $u$ satisfies the estimate \eqref{eq:Aprioriest2}. Then it is sufficient to prove that $\mathcal{F}\left[\tau;u_\epsilon,\xi_\epsilon \gamma h+(1-\xi_\epsilon )c\right]$ is weakly convergent to $\mathcal{F}\left[\tau;u,\gamma h\right]$ in $L^2(0,T;H^*_w)$. First, note that 
\[\mathcal{F}\left[\tau;u_\epsilon,\xi_\epsilon \gamma h+(1-\xi_\epsilon )c\right]=\dot{u}_\epsilon-\frac{1}{2}\sigma^2S^2u^{\prime\prime}_{\epsilon,SS}
\]
is bounded in $L^2(0,T;H^*_w)$ and then there exists an element $\tilde{\mathcal{F}}\in L^2(0,T;H^*_w)$ such that
\[\mathcal{F}\left[\tau;u_\epsilon,\xi_\epsilon \gamma h+(1-\xi_\epsilon )c\right]\underset{L^2(0,T;H^*_w)}{\rightharpoonup}\tilde{\mathcal{F}}.
\]
 On the other hand, $\mathcal{F}\left[\tau;u_\epsilon,\xi_\epsilon \gamma h+(1-\xi_\epsilon )c\right]$ is bounded from above by the constant function $\kappa$. Let $v\in L^2(0,T;H^1_w)$ be some arbitrary nonnegative function. Then Fatou's lemma implies
\begin{align}
	\left\langle \kappa-\tilde{\mathcal{F}},v\right\rangle&=\lim_{\epsilon\rightarrow0}\left(\kappa-\mathcal{F}\left[\cdot;u_\epsilon,\xi_\epsilon \gamma h+(1-\xi_\epsilon )c\right],v\right)_{L^2(0,T;L^2_w)}\nonumber\\
	&\geq \left\langle \kappa-\mathcal{F}\left[\cdot;u,\gamma h\right],v\right\rangle\geq 0,
\end{align}
i.e. 
\[\mathcal{F}\left[\cdot;u,\gamma h\right]\in L^2(0,T;H^*_w) \text{ and } \mathcal{F}\left[\cdot;u,\gamma h\right]\geq \tilde{\mathcal{F}}. \] 

Finally, we prove that in fact 
\begin{equation}
\mathcal{F}\left[\cdot;u,\gamma h\right]\equiv\tilde{\mathcal{F}},\text{ i.e., }\left\langle \mathcal{F}\left[\cdot;u,\gamma h\right],v\right\rangle=\left\langle \tilde{\mathcal{F}},v\right\rangle \quad \forall v\in L^2(0,T;H^1_w).
\label{eq:FequivFtilde}
\end{equation}
First, observe that, $v_\epsilon:=\xi_\epsilon v\rightarrow v$ as $\epsilon\rightarrow0$ in $L^2(0,T;H^1_w)$. Hence, it is sufficient to prove \eqref{eq:FequivFtilde} for functions $v$  vanishing outside a set of the form $I\times[0,T]$ where $I\subset(0,+\infty)$ is some closed interval. According to estimate \eqref{eq:Aprioriest2} and Lemma \ref{lm:Est} (applied to the interval $I$) the functions $u_\epsilon$ and $u$ are uniformly bounded on $I\times[0,T]$. Then 
\[\left\langle \mathcal{F}\left[\cdot;u,\gamma h\right],v\right\rangle=\lim_{\epsilon\rightarrow0}\left(\mathcal{F}\left[\cdot;u_\epsilon,\xi_\epsilon \gamma h+(1-\xi_\epsilon )c\right],v\right)_{L^2(0,T;L^2_w)}=\left\langle \tilde{\mathcal{F}},v\right\rangle.
\]

\noindent \textbf{Step 3.} \textit{Let $h\in H^1_w$. Then there exists a weak solution $u$ to the initial value problem \eqref{eq:psintdiff}. In addition, the inequality \eqref{eq:Aprioriest2} holds.} 

Consider a sequence  of problems with initial condition \[u_N(S,0)=\max \left\{\gamma h(S), -N\right\}, \quad N=1,2,\ldots.\] Then the corresponding solutions $u_N$ form a decreasing sequence due to the comparison principle and Lemma \ref{lm:Est}. Moreover, the pointwise limit $\lim_{N\rightarrow\infty} u_N(S,\tau)$  is finite for any $(S,\tau)$ since the inequality \eqref{eq:Aprioriest2} holds for each function $u_N$. Then the proof follows similar arguments as in Step 2. 
\end{proof}
Finally, note that the uniqueness of the weak solution is a consequence of the comparison principle. More precisely, we have
the following corollary.
\begin{corollary} Assume that $h\in H^1_w$. Then there exists a unique weak solution $u\in W(0,T)\cap L^\infty(0,T,H^1_w)$ to 
 the initial value problem \eqref{eq:psintdiff}. Moreover, the estimate \eqref{eq:Aprioriest2} holds with a constant $C>0$ independent of $u$.
\end{corollary}

{\bf{  Acknowledgement  }}

The research is supported by the European Union under Grant Agreement number 304617 (FP7 Marie Curie Action Project
Multi-INT STRIKE - Novel Methods in Computational Finance). The second author is also supported by Bulgarian National Fund
of Science under Project I02/20-2014.

\end{document}